\documentclass[a4paper,onecolumn,11pt]{quantumarticle}
\pdfoutput=1
\usepackage[utf8]{inputenc}
\usepackage[english]{babel}
\usepackage[T1]{fontenc}
\usepackage{amsthm}
\usepackage{amsmath}
\usepackage{hyperref}
\usepackage{tikz}
\usetikzlibrary{quantikz}
\usepackage{amssymb}
\usepackage{braket}
\usepackage{algorithm}
\usepackage{algpseudocode}
\usepackage{qtree}
\usepackage{caption}
\usepackage{lipsum}
\newtheorem{theorem}{Theorem}

\newtheorem{Lemma}[theorem]{Lemma}

\DeclareCaptionFormat{algor}{%
\hrulefill\par\offinterlineskip\vskip1pt%
\textbf{#1#2}#3\offinterlineskip\hrulefill}
\DeclareCaptionStyle{algori}{singlelinecheck=off,format=algor,labelsep=space}
\begin{document}

\title{Quantum Enhanced Pattern Search Optimization}
\author{Colton Mikes, Ismael R. de Farias Jr., David Huckleberry Gutman, Victoria E. Howle}
\affiliation{Texas Tech University, Lubbock,
TX 91125, USA}
\maketitle

\begin{abstract}
  This paper introduces a quantum-classical hybrid algorithm for
  generalized pattern search (GPS) algorithms. We introduce a quantum search step
  algorithm using amplitude amplification, which reduces the number of oracle calls 
  needed during the search step from $\mathcal{O}\left(N\right)$ classical calls to 
  $\mathcal{O}\left(\sqrt{N}\right)$ quantum calls. 
  This work addresses three fundamental issues with using a quantum
  search step with GPS.  First we address the need to 
  mark an improved mesh point, a requirement of the amplitude amplification
  algorithm. Second, we introduce a modified version of the amplitude
  amplification algorithm QSearch, which is guaranteed to terminate using a
  finite number of iterations. Third, we avoid disrupting the GPS algorithm's
  convergence by limiting the quantum algorithm to the search step. 
\end{abstract}

\section{Introduction}

Traditional optimization algorithms rely on derivatives to obtain their results.
However, in many cases derivative information may be unavailable. Even if 
derivatives of some order exists, it may be impractical or prohibitively expensive to
calculate any of them. For example, it can be impractical to obtain gradient information when the objective function is computed using a
black-box simulation. Similarly, approximation methods such as finite differences become
prohibitively expensive for objective functions that are noisy or expensive to
calculate. Even for functions constructed recursively from elementary functions
using composition and arithmetic operations, gradient computation may incur onerous costs in
time and memory. Automatic differentiation can also be used to obtain gradient
information in some cases, but is not always practical since it needs
the source code for the function to be available.
In these situations we turn to derivative-free optimization (DFO) algorithms. 
DFO algorithms seek to solve a given
optimization problem without the use of differential information of any order. Instead they
rely on calls to a given oracle to evaluate the objective function. These calls
are often expensive and become the dominant cost of the algorithm. 
\par
In this work, we are concerned with a class of DFO algorithms known as 
Generalized Pattern Search (GPS). These algorithms generate a sequence of 
iterates with non-increasing objective function values by systematically
selecting points to be evaluated and considered for the next iterate. This 
process is split into two steps, the \textit{search step}, and the \textit{poll step}. 
The only difference between them is the points which are evaluated. The search
step is designed to allow for a great amount of flexibility in selecting which
points to consider, while the poll step considers only those points which are
needed to guarantee the algorithm will converge. A typical iteration will begin
with the search step. The poll step is invoked only if the search step 
fails to produce an iterate with a reduced objective function value. Since the poll step is 
not always necessary, the cost of implementing a GPS algorithm is often dominated by the search step.  
However, the various convergence results introduced by Audet and Dennis make no
assumptions about the search step \cite[Theorems 3.7, 3.9, 3.14 and Corollaries
3.11, 3.12, 3.16]{Audet}.
\par
The main contribution of this paper is a quantum algorithm which is able to
perform the search step using $\mathcal{O}\left(\sqrt{N}\right)$ calls to a
quantum oracle, compared to the $\mathcal{O}\left(N\right)$ needed using a
classical oracle. As in \cite{Gilyn},\cite{Jordan}, and \cite{Bernstein}, we 
follow the standard practice of assuming that the
difference in computational cost of a quantum and classical oracle call is
negligible. Our algorithm may be separated into two components: a modified version of the
amplitude amplification algorithm QSearch \cite{Brassard}, and a quantum
algorithm, which uses a quantum oracle to prepare the necessary inputs for our
modified QSearch algorithm. 
\par
The idea of using amplitude amplification to accelerate GPS algorithms was 
previously touched on by Arunachalam in \cite{Aru}. However,
in that work they failed to address the problems of marking
points that provide a reduction in objective function value,
how to guarantee that amplitude amplification will terminate in a finite number of steps,
and how to ensure convergence with the probability of error introduced by the quantum step. 
In this paper, we address all three of these problems. Our algorithm
uses a quantum oracle as  part of a subroutine for marking the points which provide a 
reduction in objective value; it is guaranteed to terminate in a finite number of steps
and preserves the original algorithm's convergence properties by performing the poll step on a classical computer.  
\par
The rest of this paper is organized as follows.
In section \ref{sec:preliminaries}, we review GPS
algorithms, quantum computing, and how to represent our optimization problem
using a quantum computer. In section \ref{sec:qsearch}, 
we introduce a modified version of the amplitude amplification algorithm
QSearch \cite[Theorem 3]{Brassard}. In section \ref{sec:amplitude_amplification}, we
discuss amplitude amplification and the original QSearch algorithm. In section
\ref{sec:stopping_criteria}, we introduce the modifications needed to ensure QSearch will
terminate using a finite number of iterations. In section \ref{sec:the_quantum_search_step}, 
we discuss our main product, the quantum search step algorithm. In section
\ref{sec:the_algorithm}, we develop the components needed to use our modified
QSearch algorithm to perform the search step of a GPS algorithm. In section
\ref{sec:advantages}, we discuss the advantages of using the quantum search step algorithm over
classical methods. We end the paper with some concluding remarks in section
\ref{sec:concluding_remarks}. 

\section{Preliminaries}\label{sec:preliminaries}

In this section, we review all of the technology and notation from GPS
algorithms and quantum computing required to understand this paper's main
developments. In section \ref{sec:gps}, we review GPS algorithms, following the
formulation presented in by Audet and Dennis in \cite{Audet}. In the original formulation
introduced by Torczon in \cite{Torczon}, the search and poll steps were not
treated as distinct entities. By separating an iteration into the search and
poll steps, Audet and Dennis were able to show that the search step may provide
increased flexibility, while convergence is guaranteed by the poll step. In section \ref{sec:qc}, we review some
of the fundamental components of quantum computing, with an eye toward the
ingredients necessary to represent our optimization problem on a quantum
computer.  

\subsection{Generalized Pattern Search Algorithms}\label{sec:gps}

In this section, we review the construction of GPS algorithms. These algorithms were originally
introduced by Torczon in \cite{Torczon} and expanded on by Audet and Dennis in
\cite{Audet}. Here we follow Audet and Dennis's formulation, where iterations
are separated into distinct search and poll steps, so that we may accelerate the
search step, without losing the guarantee of convergence provided by the poll
step.
\par
GPS algorithms solve the optimization problem
\begin{equation}\label{eq:opt}
  \min_{x \in \mathbb{R}^{n}} f\left(x\right)  
\end{equation}
for some $f:\mathbb{R}^{n} \rightarrow \mathbb{R}$, by generating a sequence of
iterates $\left\{x_{k}\right\}_{k=0}^{\infty}$ such that
$\left\{f\left(x_{k}\right)\right\}_{k=0}^{\infty}$ is a non-increasing
sequence.  Each iteration consists of
two phases, the \textit{search step} and the \textit{poll step}.
\par
In the search step, we evaluate $f$ at a finite number of points from a set called the
\textit{mesh}. At the $k$th iteration, the mesh is the set defined by 
\begin{equation}
M_{k} = \{x_{k} + \Delta_{k}Dz:z \in \mathbb{N}^{p}\},
\end{equation} 
where $x_{k}$ is our current iterate, $\Delta_{k}$ is a positive real number called the \textit{mesh size
parameter}, $D$ is a real valued $n \times p$ matrix whose columns form a
\textit{positive spanning set}, i.e. every vector in $\mathbb{R}^{n}$ can be
written as a non-negative linear combination of $D$'s columns. 
Additionally, $D$ satisfies the
restriction that its columns are the products of some non-singular \textit{generating
matrix} $G \in
\mathbb{R}^{n \times n}$ and integer vectors $z_{j} \in \mathbb{Z}^{n}$ for
$j=1,2,\hdots,p$. Symbolically, this restriction means $D$ takes the form $D =
\left[Gz_{1} \hdots Gz_{p}\right]$.   
At iteration $k$, we call a point $y \in M_{k}$ a \textit{improved mesh point}
if $f\left(x_{k}\right) > f\left(y\right)$. If such a point is found during the
search step, we end the iteration. Otherwise we begin the poll step.
\par
During the $k$th iteration, a subset of the columns of $D$ is selected, denoted by $D_{k}$,
which is a positive spanning set. The \textit{poll set} $P_{k}$ is formed using
the current iterate and the elements of $D_{k}$:   
\begin{equation}
  P_{k} = \{ x_{k} + \Delta_{k}d: d \in D_{k}\}.
\end{equation}
If the search step fails, we invoke a poll step, in which
we evaluate the objective function for each of the
elements in $P_{k}$ to check for an improved mesh point. If none are found, 
we call the current iterate $x_{k}$ a \textit{mesh local optimizer}.  
\par
Whenever a point $y$ is identified as an improved mesh point by either the search or poll
steps, we immediately end the iteration, set $x_{k+1} = y$,  and either increase or maintain 
the mesh size parameter. Alternatively, if $x_{k}$ is determined to be a mesh
local optimizer, then we set $x_{k+1} = x_{k}$ and shrink the mesh size
parameter. We refer interested readers to \cite{Audet} for a more detailed 
presentation of this version of the GPS algorithm.   

\subsection{Quantum Computing}\label{sec:qc}

In this section, we review the relevant basics of quantum computing, and discuss how to represent our optimization problem \eqref{eq:opt} on a quantum computer. As its name implies, 
a quantum computer is a computational device that achieves
its results by manipulating quantum mechanical systems. At any given time, the
state of a quantum computer can be described by a unit vector, called the
\textit{state vector}, in a complex valued inner product space. We store this
vector using \textit{qubits}, the fundamental unit of quantum computation. The
state of an arbitrary qubit may be written as $\ket{\phi} = \alpha\ket{0} + \beta\ket{1}$,
for $\alpha,\beta \in \mathbb{C}$, where $\left\{\ket{0},\ket{1}\right\}$ is the standard basis for
$\mathbb{C}^{2}$, referred to as the \textit{computational basis} for a single
qubit, denoted here by $\mathcal{C}\left(1\right)$. 
For those more familiar with non-physical mathematical notation, we
note that $\ket{0} = \left(1,0\right)^{T}$ and $\ket{1} = \left(0,1\right)^{T}$. 
Since $\ket{\phi}$ must be a unit vector, we require that
$\|\alpha\|^{2} + \|\beta\|^{2} = 1$. We call $\alpha$ and $\beta$ the
\textit{amplitudes} of $\ket{\phi}$. 
\par
For quantum computers that contain multiple distinct qubits,
the state vector may be obtained by taking the Kronecker product of the
qubits. The computational basis for a quantum computer with $2$ qubits,
$\mathcal{C}\left(2\right)$, is the Kronecker products of the elements in 
$\mathcal{C}\left(1\right)$ with each other. In general, the
computational basis of a quantum computer with $N$ qubits,
$\mathcal{C}\left(N\right)$, is the Kronecker products of all the elements in
$\mathcal{C}\left(N-1\right)$ and $\mathcal{C}\left(1\right)$.  
For two arbitrary qubits $\ket{\phi}$ and $\ket{\theta}$, we write the
composite state as $\ket{\phi}\ket{\theta} = \ket{\phi\theta}$. For a
quantum computer with $N$ qubits, all in the same arbitrary state $\ket{\phi}$,
we denote the state vector as $\ket{\phi}^{\otimes N}$. Additionally, when we
have a quantum computer with $N$ qubits, all in the same state from
$\mathcal{C}\left(1\right)$ , we often
omit the superscript entirely and write the state as $\ket{0}$ or $\ket{1}$, so
long as the number of qubits is clear from context. We define a given set of qubits as a 
\textit{quantum register}. A set of qubits is said to be \textit{entangled} if 
their state vector cannot be written as the Kronecker product of single qubit states.
\textit{Quantum logic gates} are modeled as unitary matrices, which may be applied to the state vector.
\par
An import aspect of quantum mechanics is that we cannot view the state vector directly.
All attempts to do so will produce one of the basis vectors instead. This action
is known as \textit{measurement}. The results of measurement are determined by
the amplitudes of the state vector. If more than one of the amplitudes are
non-zero, we say that the state vector is in a \textit{superposition} of the
corresponding states.  For example, the probability of
obtaining $\ket{0}$ upon measuring the qubit
$\ket{\phi} = \alpha\ket{0} + \beta\ket{1}$ is $\|\alpha\|_{2}^{2}$, similarly,
the probability of obtaining $\ket{1}$ is $\|\beta\|_{2}^{2}$. If $\alpha$ and
$\beta$ are both non-zero, we say that $\ket{\phi}$ is in a superposition of
$\ket{0}$ and $\ket{1}$. A \textit{quantum algorithm} is defined as a unitary matrix applied 
to the state vector, possibly followed by measurement. If no measurement occurs,
we call the quantum algorithm \textit{measurement-free}. For quantum computers with multiple 
qubits, measurement is performed on each qubit individually. In some
circumstances, it is advantageous to only measure a select few of the qubits. For example, algorithms
such as the swap test \cite{Buhrman} rely on measuring a single qubit, and using
the results of that measurement to obtain information on the resulting state
vector.
\par
We now discuss how to represent our optimization problem on a quantum computer.
We represent an arbitrary real number, up to $d$ bits of precision, using
a fixed-point \textit{Two's Complement} binary representation. The Two's
Complement representation is among the most common ways to represent signed
binary numbers. In this representation, the most significant bit represents the
sign of a given binary number. If this bit is $0$, the number is positive. If it
is $1$, it is negative. The negation of a given number is obtained by
flipping each of the bits in its representation and then adding $1$. 
\par
For $x \in \mathbb{R}^{n}$, we write $x$ as $d n$ bit binary string, where each element
is represented up to $d$ bits of precision, and the $kth$ set of $d$ bits in the
string is the $d$ bit fixed-point representation of the $kth$ element of $x$. Thus $x$ is encoded into $d n$
qubits, where the state of the $k$th qubit is determined by the $k$th bit of $x$. For example,
if $x=0110$, then $\ket{x} = \ket{0110}$. To evaluate the objective function $f$, 
we assume access to a quantum oracle $F$, such that 
$F\ket{x}\ket{0}^{\otimes n \cdot d} \rightarrow \ket{x}\ket{f\left(x\right)}$. 
We may use a single call to the oracle $F$ to evaluate the objective function for a set of points
$\left\{x_{j}\right\}_{j=0}^{m-1}$. This is achieved by first initializing a
quantum register in a superposition of the points and applying $F$,
\begin{equation*}
F\sum\limits_{j=0}^{m-1}\ket{x_{j}}\ket{0} \rightarrow 
\sum\limits_{j=0}^{m-1}\ket{x_{j}}\ket{f\left(x_{j}\right)}. 
\end{equation*}
However, since the
function evaluations are stored as a superposition, we cannot view them
directly. In the following sections, we will introduce the tools necessary for
using $F$ to identify reduced objective function values.  
\section{Modified QSearch}\label{sec:qsearch}

In this section, we introduce a modified version of the 
\textit{amplitude amplification} method \textit{QSearch} introduced 
in \cite[Theorem 3]{Brassard}. If a quantum computer is in a superposition of desired and
undesired states, the QSearch algorithm may be used to obtain a desired state
upon measurement. The need for our modification arises from the fact that
QSearch is not guaranteed to terminate if there are no desired states in the
superposition. The modified QSearch algorithm will terminate once the
probability of the superposition containing a desired state is sufficiently low. 
In section \ref{sec:amplitude_amplification}, we review the core
concepts of amplitude amplification and the original QSearch algorithm. 
In section \ref{sec:stopping_criteria}, we introduce a method for determining
when the probability of the superposition containing a desired result has fallen
below an arbitrary tolerance value. 
\subsection{Amplitude Amplification}\label{sec:amplitude_amplification}

In this section, we show how the amplitude amplification algorithm QSearch
\cite[Theorem 3]{Brassard} can be used to perform the search step of a GPS algorithm. We
begin with a brief review of the algorithm itself.
\par
Suppose we have a measurement-free quantum algorithm $\mathcal{A}$.
Furthermore, suppose that when $\mathcal{A}$ is applied to the initial state
$\ket{0}^{\otimes N}$, the result is a superposition of computational basis
states in $\mathcal{C}\left(N\right)$ some of which are called \textit{desired} states ($\mathcal{D}$) and others which are called \textit{undesired} states ($\mathcal{U}$). More formally, we assume that  
\begin{equation}
  \mathcal{A}\ket{0} = \sum_{\ket{j}\in\mathcal{D}}\alpha_{j}\ket{j} +\sum_{\ket{k}\in\mathcal{U}}\alpha_{k}\ket{k}
\end{equation}
where $\mathcal{D},\mathcal{U}\subseteq \mathcal{C}\left(N\right)$ with $\mathcal{D}\cap\mathcal{U}=\emptyset$ and
\[
\sum_{\ket{j}\in\mathcal{D}}|\alpha_{j}|^2 +\sum_{\ket{k}\in\mathcal{U}}|\alpha_{k}|^2=1.
\]
Amplitude amplification methods are a class of quantum algorithms designed to
increase the probability of obtaining a desired state upon
measurement. Introduced in \cite{Brassard}, the authors developed 
multiple amplitude amplification algorithms, each of which make different assumptions on 
both the number of desired states and the user's knowledge of this number. 
\par
For our purposes, we wish to be able to apply amplitude amplification when there
may be an arbitrary (possibly $0$) number of desired states, and the user has no
knowledge of this number. The assumptions made by QSearch are the closest to
satisfying these requirements. Qsearch assumes there is at least
one desired state, but that the actual number of desired states is unknown. 
It works by constructing the measurement-free quantum algorithm $Q$, which can be
repeatedly applied to increase the probability of obtaining a desired
state upon measurement. We define $Q$ as $Q =
-\mathcal{A}S_{0}\mathcal{A}^{-1}S_{\chi}$, where 
\begin{itemize}
  \item $\mathcal{A}^{-1}$ is the quantum algorithm obtained by taking the inverse of $\mathcal{A}$. 

  \item $S_{0}$ is the measurement-free quantum algorithm, which acts as the identity for all
    computational basis states other than $\ket{0}^{\otimes N}$, and $S_{0}\ket{0}^{\otimes N} =
        -\ket{0}^{\otimes N}$. Thus for $\ket{x} \in \mathcal{C}\left(N\right)$ 
        \begin{equation*}
          S_{0}\ket{x} =
          \begin{cases}
           -\ket{x} & \text{if } \ket{x} = \ket{0}^{\otimes N} \\
            \ket{x} & \text{if } \ket{x} \neq \ket{0}^{\otimes N} \\
          \end{cases}
        \end{equation*}

  \item $S_{\chi}$ is the measurement-free quantum algorithm that flips the sign of desired
    computational basis states. Thus for $\ket{x} \in \mathcal{C}\left(N\right)$ 
    \begin{equation*}
      S_{\chi}\ket{x} = 
      \begin{cases}
        -\ket{x} & \text{if } \ket{x}\in\mathcal{D}\\
         \ket{x} & \text{if } \ket{x}\in\mathcal{C}\left(N\right)\setminus\mathcal{D}\\
      \end{cases}
    \end{equation*}
\end{itemize}
With our main ingredient in hand, the quantum algorithm $Q$, we are equipped to
formally state the QSearch algorithm.
\begin{center}
  \captionof{algorithm}{QSearch}\label{alg:qsearch}
  \begin{flushleft}
    \textbf{Input:} A measurement-free quantum algorithm $\mathcal{A}$, 
    which operates on $N$ qubits, $S_{0}$ and $S_{\chi}$ as defined above,  
    and a quantum state composed of $n$ qubits initialized to the state $\ket{0}^{\otimes N}$.   
      \end{flushleft}
  \begin{algorithmic}[1]
   \State Set $l = 0$ and let $c$ be any constant such that $1 < c < 2$. 
    \State Apply $\mathcal{A}$ to $\ket{0}^{\otimes N}$ and measure. If the result of
    measurement is a desired state, then return the state. 
    \While{A desired state has not been returned}{}
    \State Increase $l$ by 1 and set $M = \lceil c^{l} \rceil$.
    \State Initialize a register of appropriate size to the state
    $\mathcal{A}\ket{0}^{\otimes N}$.  
    \State Set $j$ to be an integer in $\left[1,M\right]$ uniformly at random.
    \State Apply the algorithm $Q$ to the state $\mathcal{A}\ket{0}^{\otimes N}$
    $j$ times.
    \State Measure the state. If it is desired, return the state. 
    \EndWhile
  \end{algorithmic}
  \begin{flushleft}
  \textbf{Output:} A desired state.  
  \end{flushleft}
 \rule[15pt]{\textwidth}{0.5pt}
\end{center}

The following theorem from \cite[Theorem 3]{Brassard} elaborates on the number of 
applications of $\mathcal{A}$ and $\mathcal{A}^{-1}$ that will be required by QSearch.
For the sake of concision, we refer the reader to \cite{Brassard} for its proof.

\begin{theorem}\label{thm:qsearch}
  Let $\mathcal{A}$ be a measurement-free quantum algorithm, such that 
  $\mathcal{A}\ket{0}^{\otimes N}$ may be written as the
  superposition of $N$ states, $t$ of which are desired. Exactly one of the
  following cases holds based on the value of t:
  \begin{itemize}

    \item $t > 0$: QSearch will return a desired state using
      $\Theta\left(\sqrt{\frac{N}{t}}\right)$ applications of $\mathcal{A}$ and
      $\mathcal{A}^{-1}$.

    \item $t=0$: QSearch will fail to terminate.

   \end{itemize}
\end{theorem}
Theorem \ref{thm:qsearch} guarantees that if a successful state exists among
those being searched, it will be found using $\Theta\left(\sqrt{N}\right)$
applications of $\mathcal{A}$. However, if such a state does not exist, the
algorithm does not terminate. To apply QSearch in our setting, we must ensure
that it will terminate, regardless of if a desired state exists or not.
\subsection{Adding a Stopping Criteria to QSearch}\label{sec:stopping_criteria}

In this section, we introduce a stopping criteria to be used in our
modified QSearch. This criteria will allow for an arbitrarily low probability of
terminating the algorithm when a desired state exists. 
\par
Suppose we have access to a measurement-free quantum algorithm $\mathcal{A}$
such that $\mathcal{A}\ket{0}^{\otimes N}$ may be written as a superposition of
$N$ states, $t$ of which are desired. As stated in theorem \ref{thm:qsearch}, if
$t=0$, then QSearch will fail to terminate. Unless we measure each of the $N$
states, it is impossible to know for certain that $t=0$. However, we can bound
the probability of $t>0$ using our knowledge of $N$ and the number of
iterations of QSearch that have been performed. 
Suppose $0 < \frac{t}{N} < \frac{3}{4}$, and $M$ is some integer such that $M >
\sqrt{N}$. Then as shown in the proof of 
\cite[Theorem 3]{Brassard}, if we apply $Q$, $j$ times, where the integer $j$ is chosen from 
$\left[1,M\right]$ uniformly at random, then the probability of
obtaining a desired state is bounded below by
\begin{equation}
  \frac{1}{2}\left(1 - \frac{1}{2M\sqrt{\frac{t}{N}}}\right).
\end{equation}
We use this knowledge to prove the following lemma: 
\begin{Lemma}\label{thm:stopping}
  Suppose we have access to a measurement-free quantum algorithm $\mathcal{A}$
  such that $\mathcal{A}\ket{0}^{\otimes N}$ may be written as a superposition
  of $N$ states, $t$ of which are desired. Let $M \in \mathbb{Z}, M> \sqrt{N}$,
  and $j \in \mathbb{Z}$ is an integer chosen from $\left[1,M\right]$ uniformly 
  at random. If $0 < \frac{t}{N} < \frac{3}{4}$ and we have applied $u$ iterations of
  QSearch, each failing to produce a desired state, then the probability
  of failing to find a desired state on the next iteration of QSearch
  is bounded above by $\left(\frac{3}{4}\right)^{u}$.
\end{Lemma}
\begin{proof}
In the course of proving Theorem \ref{thm:qsearch} (\cite[Theorem 3]{Brassard}), \cite{Brassard} shows that if we apply $Q$, $j$ times, then the probability of
obtaining a desired state is bounded below by $\frac{1}{2}\left(1 - 
\frac{1}{2M\sqrt{\frac{t}{N}}}\right)$. Since we do not know the value
of $t$ ahead of time, we assume the worst case of $t=1$ and reduce this bound to obtain our result. 
\begin{equation*}  
\begin{split}
\frac{1}{2}\left(1-\frac{1}{2M\sqrt{\frac{t}{N}}}\right) &=  \frac{1}{2}\left(1 -
\frac{\sqrt{N}}{2M\sqrt{t}}\right) \\
& \geq \frac{1}{2}\left(1 - \frac{\sqrt{N}}{2M}\right).
\end{split}
\end{equation*}
For $M > \sqrt{N}$ we have
\begin{equation*}
\frac{1}{2}\left(1 - \frac{\sqrt{N}}{2M}\right) > \frac{1}{4}.
\end{equation*}
Thus the probability of failing to find a desired state is bounded above by
$\frac{3}{4}$. Since each iteration of QSearch is independent, the probability
of failing to find a desired state during $u$ iterations of QSearch is bounded above by
$\left(\frac{3}{4}\right)^{u}$.
\end{proof}
Using Lemma \ref{thm:stopping}, we can determine when the probability of missing
an existing desired solution has fallen below an arbitrary tolerance value. We
now have the tools necessary for the introduction of our modified QSearch
algorithm.
\begin{center}
  \captionof{algorithm}{Modified QSearch}\label{alg:mod_qsearch}
  \begin{flushleft}
    \textbf{Input:} A measurement-free quantum algorithm $\mathcal{A}$, 
    which operates on $N$ qubits, $S_{0}$ and $S_{\chi}$ as defined above, a
    tolerance value $\tau>0$,  
    and a quantum state composed of n qubits initialized to the state $\ket{0}^{\otimes N}$.    
      \end{flushleft}
  \begin{algorithmic}[1]
   \State Set $l=0$ and let $c$ be any constant such that $1<c<2$. 
    \State Set $u$ = 0.
    \State Apply $\mathcal{A}$ to $\ket{0}^{\otimes N}$ and measure. If the result of
    measurement is a desired state, then return the state. 
    \While{A desired state has not been returned and 
           $u<\frac{\ln\left(\tau\right)}{\ln\left(\frac{3}{4}\right)}$}{}
    \State Increase $l$ by 1 and set $M = \lceil c^{l} \rceil$.
    \If{$M > \sqrt{N}$}
    \State Set  $u = u + 1$.
    \EndIf
    \State Initialize a register of appropriate size to the state
    $\mathcal{A}\ket{0}^{\otimes N}$.  
    \State Set $j$ to be an integer in $\left[1,M\right]$ uniformly at random.
    \State Apply the algorithm $Q$ to the state $\mathcal{A}\ket{0}^{\otimes N}$
    $j$ times.
    \State Measure the state. If it is desired, return the desired state. 
    \EndWhile
    \State If a desired state was not found, return the failure to find a
    desired state.
  \end{algorithmic}
  \begin{flushleft}
  \textbf{Output:} A desired state or failure to find a desired state.  
  \end{flushleft}
 \rule[15pt]{\textwidth}{0.5pt}
\end{center}
With our modified QSearch algorithm introduced, we now show how it may be
used to perform the search step of a given GPS algorithm.
\section{The Quantum Search Step}\label{sec:the_quantum_search_step}

In this section, we introduce our main product, the quantum search step
algorithm. This algorithm completes the search step of a given
GPS algorithm using $\mathcal{O}\left(\sqrt{N}\right)$ quantum oracle calls, compared to
the $\mathcal{O}\left(N\right)$ classical oracle calls. Here we follow the
standard assumption that the computational cost between quantum and classical
oracle calls is negligible \cite{Gilyn},\cite{Jordan},\cite{Bernstein}. Let $X =
\left\{x_{j}\right\}_{j=0}^{N-1}$. Recall that we represent $x\in X$
as a $d n$ binary string, where the $k$th set of $d$ bits represents the
$k$th element of the point. For a given objective function $f:\mathbb{R}^{n}
\rightarrow \mathbb{R}$, we assume access to a quantum oracle $F$, which is able
to implement the transformation 
\begin{equation*}
F\sum\limits_{j=0}^{N-1}\ket{x_{j}}\ket{0} \rightarrow
\sum\limits_{j=0}^{N-1}\ket{x_{j}}\ket{f\left(x_{j}\right)}. 
\end{equation*}
In section \ref{sec:the_algorithm} we show how to implement the operator $Q$ needed to use the modified QSearch
algorithm before formally stating the quantum search step algorithm. In section
\ref{sec:advantages}, we discuss the advantages of using the quantum search step algorithm over
classical methods, how the quantum search step preserves the original
convergence properties of a given GPS algorithm, and why we choose not to use
the modified QSearch algorithm to perform the poll step.

\subsection{Preparing the Operator Q}\label{sec:the_algorithm}

Recall that our goal is to use the modified QSearch algorithm to search for an
improved mesh point among a set of points from a mesh, as described in section \ref{sec:gps}. This
application requires that we construct the operator $Q = \mathcal{A}S_{0}\mathcal{A}^{-1}S_{\chi}$, 
which is equivalent to implementing the operators 
$\mathcal{A}$, $S_{0}$, $\mathcal{A}^{-1}$, and $S_{\chi}$.
In this section, we will show how to implement the following components of $Q$:
\begin{itemize}
  \item $\mathcal{A}$, such that $\mathcal{A}\ket{0}^{\otimes N}$ produces a superposition 
        of the points we wish to search.  

  \item $S_{0}$ is the measurement-free quantum algorithm, which acts as the identity for all
        computational basis states other than $\ket{0}^{\otimes N}$, and $S_{0}\ket{0}^{\otimes N} =
        -\ket{0}^{\otimes N}$. Thus for $\ket{x} \in \mathcal{C}\left(N\right)$ 
        \begin{equation*}
          S_{0}\ket{x} =
          \begin{cases}
           -\ket{x} & \text{if } \ket{x} = \ket{0}^{\otimes N} \\
            \ket{x} & \text{if } \ket{x} \neq \ket{0}^{\otimes N} \\
          \end{cases}
        \end{equation*}

  \item $\mathcal{A}^{-1}$ is the quantum algorithm obtained by taking the inverse of $\mathcal{A}$. 

  \item $S_{\chi}$ is the measurement-free quantum algorithm that flips the sign of
        computational basis states corresponding to improved mesh points. 
        Thus for $\ket{x} \in \mathcal{C}\left(N\right)$ 
        \begin{equation*}
        S_{\chi}\ket{x} = 
        \begin{cases}
          -\ket{x} & \text{if } x \text{ is an improved mesh point}\\
          \ket{x} & \text{if } x \text{ is not an improved mesh point}\\
        \end{cases}
        \end{equation*}
\end{itemize}
$S_{\chi}$ presents the greatest challenge for constructing $Q$ since it requires that we be able to
mark improved mesh points. We do this by designing $\mathcal{A}$ to operate on
three quantum registers. The first two are the same as those used by our oracle
$F$. The third register will be used to compare the objective function value at
the current iterate and the points we wish to search.
\par
Let $X = \left\{x_{j}\right\}_{j=0}^{N-1}$ be a set of $N$ points chosen 
from the mesh, $f\left(x_{j}\right) = f_{j}$ be the objective function evaluated at the
point $x_{j}$, and $x_{k}$ be the current iterate. We will design $\mathcal{A}$ such that
\begin{equation}
  \mathcal{A}\ket{0} =
  \frac{1}{\sqrt{N}}\sum\limits_{j=0}^{N-1}\ket{x_{j}}\ket{f_{j}}\ket{f_{j}
- f_{k}}. 
\end{equation}
The sign qubit of $\ket{f_{j} - f_{k}}$ is then in the state $\ket{1}$ if and only if  
$\ket{x_{j}}\ket{f_{j}}\ket{f_{j} - f_{k}}$ is a desired state. We may use this
qubit as a control qubit in our construction of $S_{\chi}$. Our algorithm $\mathcal{A}$ is then as follows: 
\begin{center}
  \captionof{algorithm}{$\mathcal{A}$}\label{alg:A}
  \begin{flushleft}
\textbf{Input:} Let $f\left(x_{k}\right) = f_{k}$ be the objective function 
    evaluated at the current iterate $x_{k}$, and let $X =
    \left\{x_{j}\right\}_{j=0}^{N-1}$ be a set of $N$ points chosen from the
    mesh.
  \end{flushleft}
  \begin{algorithmic}[1]

    \State Place the third register in the state $\ket{-f_{k}}$
           \begin{equation*}
              \ket{\phi_{0}} = \ket{0}\ket{0}\ket{-f_{k}}.
            \end{equation*}

  \State Place the first register into a uniform superposition of the $N$ points. 
         \begin{equation*}
         \ket{\phi_{1}} =
          \frac{1}{\sqrt{N}
          }\sum\limits_{j=0}^{N-1}\ket{x_{j}}\ket{0}\ket{-f_{k}}.
          \end{equation*}

\State Apply the oracle $F$ to the first two registers
\begin{equation*}
  \ket{\phi_{2}} = 
  \frac{1}{\sqrt{N}}\sum\limits_{j=0}^{N-1}\ket{x_{j}
                   }\ket{f_{j}}\ket{f_{k}}.
  \end{equation*}

\State Add the second register to the third register using the in-place signed
       qft addition algorithm introduced in \cite[Section 3.1]{Sahin}
       \begin{equation*}
       \ket{\phi_{3}} = 
   \frac{1}{\sqrt{N}}\sum\limits_{j=0}^{N-1}\ket{x_{j}
    }\ket{f_{j}}\ket{f_{j} - f_{k}}.
\end{equation*}
  \end{algorithmic}
  \begin{flushleft}
  \textbf{Output:} $\ket{\phi_{3}}$ 
  \end{flushleft}
 \rule[15pt]{\textwidth}{0.5pt}
\end{center}

In this context, the set of desired and undesired states are naturally given as
\begin{align*}
  \mathcal{D}&:=\{\ket{x_{j}
    }\ket{f_{j}}\ket{f_{j} - f_{k}}:f_{j} - f_{k}< 0\}\\
  \mathcal{U}&:=\{\ket{x_{j}
    }\ket{f_{j}}\ket{f_{j} - f_{k}}:f_{j} - f_{k}\geq 0\}.
\end{align*}
With $\mathcal{A}$ defined, $\mathcal{A}^{-1}$ is easily obtained by as the
inverse of $\mathcal{A}$. To implement $S_{0}$, we first consider the single qubit gate
defined by the unitary matrix   
\begin{equation*} U =
  \begin{bmatrix}
    -1 & 0 \\
     0 & 1 
  \end{bmatrix}.
\end{equation*}
The action of this gate on the computational basis is $U\ket{0} = -\ket{0}$ and
$U\ket{1} = \ket{1}$. Thus $U$ acts as $S_{0}$ for a single qubit. However, our
algorithm will always require more than one qubit. We may use $U$ to implement
$S_{0}$ for an arbitrary number of qubits by applying $U$ to the first qubit,
controlled on the remaining qubits being in the state $\ket{0}$.   
\par
As with $S_{0}$, we can implement $S_{\chi}$ as a controlled version of a single
qubit gate. Recall that after applying $\mathcal{A}$, the sign qubit of the
comparison register is in the state $\ket{1}$ if the state represents an improved
mesh point, otherwise it is in the state $\ket{0}$. Consider the single qubit
gate represented by the matrix
\begin{equation*} V =
  \begin{bmatrix}
    e^{i\pi} & 0 \\
    0 & e^{i\pi} 
  \end{bmatrix}.
\end{equation*}
We can implement $S_{\chi}$ by applying $V$ to the first qubit of the input register, controlled on the
sign qubit of the comparison register being in the state $\ket{1}$. Now that we
have a method of preparing all the ingredients of $Q$, we are prepared to
present the quantum search step. 
\begin{center}
  \captionof{algorithm}{The Quantum Search Step}\label{alg:quantum_search_step}
  \begin{flushleft}
    \textbf{Input:} $\mathcal{A}$, $S_{0}$, $\mathcal{A}^{-1}$, and $S_{\chi}$
    as defined above, a set of points to be searched $\left\{x_{j}\right\}_{j=0}^{N-1}$, 
    the current iterate $x_{k}$, a tolerance value $\tau_{0}$ > 0, and three
    quantum registers in the state $\ket{\psi} =
    \ket{0}\ket{0}\ket{0}$.
  \end{flushleft}
  \begin{algorithmic}[1]
    \State Apply modified QSearch to $\ket{\psi}$.
    \If{modified QSearch returns an improved mesh point, $x'$}
    \State Set $x_{k+1} = x'$ and update the mesh size parameter. 
    \Else{terminate the algorithm and return the failure to find an improved
    mesh point.} 
    \EndIf
  \end{algorithmic}
  \begin{flushleft}
    \textbf{Output:} The next iterate $x_{k+1}$, or failure to find an improved
    mesh point.
  \end{flushleft}
 \rule[15pt]{\textwidth}{0.5pt}
\end{center}
With the quantum search step in hand, we now discuss its effect on convergence
and how it compares to classical methods. 

\subsection{Convergence and the Quantum Advantage}\label{sec:advantages}

In this section, we discuss the effects modified QSearch may have on the
convergence properties of GPS algorithms as well as the advantages of the quantum search step compared
to classical methods.
\par
It is worthwhile to note that the only difference between the search step and
the poll step is the set of points that are considered. The motivation for separating
an iteration into these steps is the unique roll they each fill within a given 
GPS algorithm. The search step allows for a much wider and more varied selection
of points from the mesh. This is because it places very few
restrictions on which points are to be considered. They must each each come from
a given mesh, and the total number of points must be finite. This allows an
impressive level of flexibility when deciding how to implement the search step.
In practice, this flexibility is key to the effectiveness of GPS algorithms
\cite{Audet}. However, this flexibility also prevents the search step from being
sufficient to guarantee the algorithm will converge. This is the roll of the poll
step. In \cite[Theorems 3.7, 3.9, 3.14 and Corollaries 3.11, 3.12, 3.16]{Audet}, 
Audet and Dennis present a variety of different
convergence results, each differing in it's assumptions about the objective
function. However, while these results all assume the poll step is performed as
described in section \ref{sec:gps}. However, they make no assumptions about the
search step. Thus we can think of the poll step as an extra step to ensure the algorithm will
converge. This is why we restrict our use of modified QSearch to the search
step. While we may ensure that the probability of overlooking an improved mesh
point is arbitrarily low, we cannot ensure it is $0$ without increasing our
oracle calls to $\mathcal{O}\left(N\right)$.  
\par
Suppose we wish to consider $N$ points from the mesh. Classical methods require
$\mathcal{O}\left(N\right)$ oracle calls to complete the search step. However,
the quantum search step can search the same $N$ points using
$\mathcal{O}\left(\sqrt{N}\right)$ oracle calls. We now conclude this section by
summarizing the advantages of the quantum search step in the following theorem,
\begin{theorem}\label{thm:quantum_search_step}
  Let $X$ be a set of $N$ points chosen from the mesh, $t$ of which are improved mesh
  points. Exactly one of the following cases will hold 
  \begin{itemize}
    \item The quantum search step will return an improved mesh point using
      $\Theta\left(\sqrt{\frac{N}{t}}\right)$ calls to the quantum oracle.  

    \item The quantum search step will terminate without finding the next
      iterate.
  \end{itemize}
  If $X$ contains an improved mesh point, the probability of the quantum search
  step terminating without finding the next iterate may be set arbitrarily low.
\end{theorem}
\begin{proof}
  If the quantum search step returns an improved mesh point, Theorem
  \ref{thm:qsearch} guarantees that it will do so using $\Theta\left(\sqrt{\frac{N}{t}}\right)$ 
  calls to the quantum oracle.
  \par
  Suppose the algorithm terminates without producing the next iterate and 
  $X$ contains an improved mesh point. By Lemma \ref{thm:stopping}, the
  probability of failing to find an improved mesh point on the next iteration
  is bounded above by $\left(\frac{3}{4}\right)^{u}$, where $u$ is the number of iterations of
  modified QSearch performed after $M > \sqrt{N}$. Thus for $\tau > 0$, if 
  $u < \frac{\ln\left(\tau\right)}{\ln\left(\frac{3}{4}\right)}$, then 
  \begin{equation*}
    \begin{split}
      u &< \frac{\ln\left(\tau\right)}{\ln\left(\frac{3}{4}\right)} \\
      u ln\left(\frac{3}{4}\right) &< ln\left(\tau\right) \\
      ln\left(\left(\frac{3}{4}\right)^{u}\right) &< ln\left(\tau\right) \\
      \left(\frac{3}{4}\right)^{u} &< \tau
    \end{split}
  \end{equation*}
  and we have that the probability of the quantum search step failing to find an improved mesh point
  on the next iteration is less than $\tau$.
\end{proof}

\section{Concluding Remarks}\label{sec:concluding_remarks}

In this work, we introduced a modified version of the amplitude amplification
algorithm QSearch \cite[Theorem 3]{Brassard}. Our modified QSearch guarantees that the
algorithm will terminate using a finite number of steps, with an arbitrarily low
probability of failing to produce a desired state, if one exists. We then
showed how modified QSearch can be used to implement the quantum search step, a
quantum algorithm for performing the search step of a given 
GPS algorithm. When searching $N$ points, $t$ of which are improved mesh
points, our algorithm is able to complete the search step using
$\mathcal{O}\left(\sqrt{N}\right)$ quantum oracle calls, compared to the
$\mathcal{O}\left(N\right)$ classical oracle calls needed traditionally.
If $t>0$, our quantum search step has an arbitrarily low probability of failing
to identify an improved mesh point.
\bibliography{Quantum_Enhanced_Pattern_Search_Optimization.bib}{}
\bibliographystyle{plain}

\end{document}